\documentclass[runningheads]{llncs}
\usepackage[T1]{fontenc}
\usepackage{graphicx}
\usepackage{stmaryrd}
\usepackage{physics}
\usepackage{amssymb}
\usepackage{tikz}
\usetikzlibrary{arrows.meta,positioning,fit}
\usepackage{bbding}

\begin{document}
\title{Semantics-Based Verification of an Implemented Shor Oracle for ECDLP in Qrisp}

\titlerunning{Verifying a Shor Oracle for ECDLP}

\author{Lei Zhang\Envelope\orcidID{0000-0001-9343-3654} \and
Zhiyuan Chen}
\authorrunning{L. Zhang and Z. Chen}
\institute{Department of Information Systems, University of Maryland, Baltimore County, MD, USA,  
\email{leizhang@umbc.edu, zhchen@umbc.edu}}

%
%
\maketitle              
\begin{abstract}
Shor-style quantum algorithms for the elliptic-curve discrete logarithm problem (ECDLP) are highly sensitive to the exact semantics of their group-operation oracles. Consequently, minor implementation choices can invalidate the intended mathematical model and lead to misleading conclusions. This paper introduces a semantics-first verification perspective for an end-to-end, compilable ECDLP implementation built on Qrisp. We specify the implemented oracle at the level of program semantics, derive refinement-style verification obligations for its key components, and provide a high-level complexity argument for the resulting oracle family. A small case study highlights that (i) the core point-update primitive agrees with a classical reference on well-formed inputs, yet (ii) controlled execution may violate the expected control law under the evaluated toolchain, despite a passing trivial control sanity check. These results position semantic auditing as a practical prerequisite for trustworthy ECDLP-oriented quantum software.
\end{abstract}

\section{Introduction}\label{sec:intro}

Shor-style quantum algorithms threaten widely deployed public-key cryptosystems~\cite{ahmed2025survey,zhang2020quantum,zhang2023making}. 
Quantum programming frameworks such as Qrisp~\cite{seidel2024qrisp} provide high-level abstractions and compilation infrastructure, including support for Shor-style implementations~\cite{polimeni2025end}. 
Compared to integer factorization, Shor-style approaches to the elliptic-curve discrete logarithm problem (ECDLP) are less mature in terms of end-to-end implementations and require more substantial oracle constructions for elliptic-curve group arithmetic~\cite{albuainain2022experimental,bhatia2020efficient,larasati2021quantum}.

This paper focuses on quantum software verification for Shor-style ECDLP implementations. 
Existing work has developed formal methods and program logics for quantum programs~\cite{lewis2023formal,yu2021quantum}, including certified end-to-end verification of Shor's factoring algorithm~\cite{peng2023formally}. 
However, ECDLP implementations often rely on specialized group-arithmetic gadgets, encodings, and sign conventions, leaving a gap between the mathematical oracle specified in an algorithmic description and the oracle realized after decomposition, control construction, ancilla management, and compilation. 
This gap is critical because the ECDLP oracle determines the algebraic relation recovered from the final measurements.

We address this gap by treating the implemented ECDLP oracle as a program-semantic object rather than only as an abstract mathematical function. 
Testing alone is insufficient: small sign or representation mismatches can invalidate the intended model and yield misleading interference patterns~\cite{miranskyy2025feasibility,miranskyy2019testing,miranskyy2021testing}. 
Moreover, because large-scale fault-tolerant hardware capable of threatening modern RSA and ECC parameters is not yet available, current evaluation relies mainly on mathematical reasoning, formal analysis, and small-scale simulation.

We distinguish three layers of correctness: (i) the standalone point-update primitive, (ii) its mathematically controlled version, and (iii) the controlled operation generated by the programming framework and compiler. 
This distinction matters because a primitive may behave correctly in isolation, while its controlled realization may fail to satisfy the identity branch required by the standard control law. 
We study these layers in an end-to-end, compilable ECDLP implementation built on Qrisp~\cite{polimeni2025end}, focusing on the semantic contracts for point update and controlled oracle construction. 
Our artifacts are available on Zenodo~\cite{zhang_2026_18865803}.

This paper makes three contributions. 
First, we specify the implemented oracle as a state transformer on computational basis states. 
Second, we state refinement obligations for the in-place point-update gadget and the controlled scalar-multiplication-and-add routine. 
Third, we give a structural complexity argument showing that the oracle family is polynomial under standard assumptions on reversible modular arithmetic.

\section{Preliminaries and Formal Model}\label{sec:model}

\textbf{ECDLP and Shor-style oracles.}
Let $p$ be a prime and let $E(\mathbb{F}_p)$ be an elliptic-curve group over $\mathbb{F}_p$.
Let $G \in E(\mathbb{F}_p)$ generate a cyclic subgroup of order $r$, and let $P=\ell G$ for an unknown $\ell \in \mathbb{Z}_r$.
A Shor-style ECDLP oracle uses two $n$-qubit quantum registers with computational-basis values $x_1,x_2 \in \mathbb{Z}_{2^n}$, which index scalar multiples of $G$ and $P$.
The oracle computes a group expression, such as $x_1G+x_2P$, into an accumulator point register before the inverse quantum Fourier transform~\cite{weinstein2001implementation} extracts a linear constraint.
The exact oracle semantics matters: if the implementation realizes $x_1G+x_2P$ while the analysis assumes $x_1G-x_2P$, the induced relation and classical post-processing differ by a sign.

\textbf{Points and registers.}
Let $\mathcal{O}$ denote the identity of $E(\mathbb{F}_p)$.
The accumulator encodes a point $R \in E(\mathbb{F}_p)$ using two modular registers for affine coordinates; in the case study (Section~\ref{sec:casestudy}), $(0,0)$ encodes $\mathcal{O}$.
We write $\ket{x_1}\ket{x_2}\ket{R}$ for a computational basis state.

\textbf{Implemented oracle.}
Figure~\ref{fig:oracle-composition} summarizes the compiled oracle as two sequential scalar-multiplication-and-add updates of the accumulator, controlled bitwise by $x_1$ and $x_2$. Let 
$f_{\mathrm{impl}}(x_1,x_2) = x_1G \oplus x_2P$ 
denote the point-valued function computed into the accumulator by the implementation.
The implemented oracle is
\begin{equation}
\mathcal{O}_{\mathrm{impl}} : \ket{x_1}\ket{x_2}\ket{R} \mapsto
\ket{x_1}\ket{x_2}\ket{R \oplus x_1G \oplus x_2P},
\label{eq:oracle-impl}
\end{equation}
where $\oplus$ denotes the elliptic-curve group operation as realized by the implementation's point encoding. Substituting $P=\ell G$ gives
$f_{\mathrm{impl}}(x_1,x_2)=x_1G \oplus x_2P=(x_1+\ell x_2)G$.

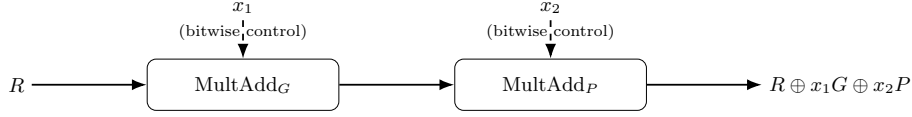
\begin{figure}[t]
\centering
\resizebox{\linewidth}{!}{%
\begin{tikzpicture}[
  font=\small,
  box/.style={draw, rounded corners, minimum height=8mm, minimum width=30mm, align=center},
  data/.style={-Latex, thick},
  ctrl/.style={-Latex, dashed, thick},
  lab/.style={align=left}
]
\node[box] (mG) {$\mathrm{MultAdd}_G$};
\node[box, right=18mm of mG] (mP) {$\mathrm{MultAdd}_P$};
\node[lab, right=18mm of mP] (out) {$R \oplus x_1G \oplus x_2P$};

\node[lab, left=18mm of mG] (R) {$R$};
\draw[data] (R.east) -- (mG.west);
\draw[data] (mG.east) -- (mP.west);
\draw[data] (mP.east) -- (out.west);

\node[lab, above=6mm of mG] (x1) {$x_1$};
\draw[ctrl] (x1.south) -- (mG.north);
\node[font=\scriptsize, above=5pt of mG] {(bitwise control)};

\node[lab, above=6mm of mP] (x2) {$x_2$};
\draw[ctrl] (x2.south) -- (mP.north);
\node[font=\scriptsize, above=5pt of mP] {(bitwise control)};

\end{tikzpicture}
}
\caption{Oracle composition in the implementation. The accumulator $R$ is updated by two scalar-multiplication-and-add routines, controlled bitwise by $x_1$ and $x_2$, yielding $R \oplus x_1G \oplus x_2P$. Here $\oplus$ denotes the elliptic-curve group operation under the implementation's point encoding, not bitwise XOR.}
\label{fig:oracle-composition}
\end{figure}

\textbf{Refinement notion.}
Let $C$ be a circuit acting on registers $(x_1,x_2,R)$.
We write $\llbracket C\rrbracket$ for its denotation on computational basis states.
We say $C$ \emph{refines} $\mathcal{O}_{\mathrm{impl}}$ (written $C \sqsubseteq \mathcal{O}_{\mathrm{impl}}$) if, for all basis states $\ket{x_1}\ket{x_2}\ket{R}$ satisfying a well-formedness predicate $\mathsf{WF}(x_1,x_2,R)$, the action of $C$ on the point register equals~\eqref{eq:oracle-impl}, and all ancillas are returned to their initial states.

\textbf{Well-formedness.}
We parameterize our specifications by a predicate $\mathsf{WF}$ that excludes exceptional inputs for the affine gadget, such as division by zero in slope computation, adding inverse points, or invalid intermediate encodings of $\mathcal{O}$.
In this pilot study, $\mathsf{WF}$ is treated as an explicit precondition; handling exceptional cases is left to future work.

\section{Semantic Verification Obligations and Complexity}\label{sec:verification}

This section states the refinement obligations induced by the implemented oracle and derives a structural bound on the oracle's circuit size. The obligations reflect the three-layer distinction introduced in Section~\ref{sec:intro}: standalone primitive correctness, controlled-operation correctness, and correctness after framework-level generation and compilation.

\subsection{In-place point update (correctness obligation)}\label{sec:add-obligation}

Let $\mathrm{AddInpl}_G$ denote the implementation's in-place point-update gadget for adding a fixed base point $G$ over modulus $p$.
Its intended behavior is
\[
\ket{R} \mapsto \ket{R \oplus G}.
\]
Here \emph{in-place} means that the accumulator register storing $R$ is overwritten with $R \oplus G$, up to ancillas that are uncomputed, rather than writing the result to a separate output register.
The first verification obligation is a partial-correctness statement for this standalone primitive.

\begin{lemma}[In-place point update (partial correctness)]
\label{lem:addinpl}
For any basis state $\ket{R}$ such that $\mathsf{WF}(R,G)$ holds,
$\llbracket \mathrm{AddInpl}_G \rrbracket(\ket{R}) = \ket{R \oplus G}$,
and the gadget returns all ancillas to $\ket{0}$.
\end{lemma}

Lemma~\ref{lem:addinpl} is an algebraic verification task: the gadget computes
slope-like values using reversible modular arithmetic, including the fixed-round
Kaliski-style inversion routine used in Qrisp ECDLP~\cite{polimeni2025end},
then updates affine coordinates and uncomputes temporaries. We treat it as an
explicit obligation for future mechanization or symbolic validation.

\subsection{Controlled scalar multiplication-and-add}\label{sec:multadd}

Let $\mathrm{MultAdd}_G$ denote the scalar-multiplication-and-add routine that processes an $n$-qubit register $k$ bitwise and conditionally applies $\mathrm{AddInpl}_{2^iG}$, where the doubles $2^iG$ are generated classically. Its intended effect is
\[
(k,R)\mapsto (k,\, R \oplus kG),
\]
i.e., it accumulates $kG$ into $R$ via conditional additions of precomputed doubles.

\begin{lemma}[Scalar multiplication-and-add (refinement)]
\label{lem:multadd}
Assume Lemma~\ref{lem:addinpl} holds for all constants in the doubling schedule, the controlled invocations satisfy the standard control law, and $\mathsf{WF}$ holds for all intermediate accumulator states. Then, for any
basis state $\ket{k}\ket{R}$ with $k\in\mathbb{Z}_{2^n}$,
\[
\llbracket \mathrm{MultAdd}_G \rrbracket(\ket{k}\ket{R}) = \ket{k}\ket{R \oplus kG}.
\]
\end{lemma}

\begin{proof}
By induction on the loop index $i$. After $i$ iterations, the accumulator is
$R \oplus \sum_{j<i} k_j(2^jG)$.
At iteration $i$, the controlled invocation of $\mathrm{AddInpl}_{2^iG}$ adds
$2^iG$ iff $k_i=1$ and otherwise acts as identity.
Thus the invariant is extended to $j\leq i$.
After $n$ iterations, the accumulator is $R\oplus \sum_{j<n}k_j(2^jG)=R\oplus kG$.
\end{proof}

\subsection{Oracle refinement and polynomial complexity}\label{sec:complexity}

By composing two instances of Lemma~\ref{lem:multadd} (for $G$ and $P$), the
implementation refines $\mathcal{O}_{\mathrm{impl}}$ in~\eqref{eq:oracle-impl}
under $\mathsf{WF}$ and the stated controlled-correctness assumptions.

We now argue that the implementation's oracle circuit family is polynomial in
the field-size parameter $m=\lceil \log_2 p\rceil$.

\begin{lemma}[Structural cost bound]
\label{lem:structcost}
Let $\mathrm{Cost}_{\mathrm{Add}}(m)$ be a cost measure (e.g., gate count or
$T$-count) for $\mathrm{AddInpl}_G$ on $m$-bit modular registers.
Then $\mathrm{MultAdd}_G$ has cost at most $n\cdot \mathrm{Cost}_{\mathrm{Add}}(m)+O(n)$,
and the oracle $\mathcal{O}_{\mathrm{impl}}$ has cost at most
$2n\cdot \mathrm{Cost}_{\mathrm{Add}}(m)+O(n)$.
\end{lemma}

\begin{proof}
$\mathrm{MultAdd}_G$ iterates exactly $n$ times and performs at most one
controlled invocation of $\mathrm{AddInpl}$ per iteration; the remaining work is
constant overhead and classical constant updates. The oracle composes two such
calls.
\end{proof}

\begin{theorem}[Polynomial oracle family (under standard arithmetic assumptions)]
\label{thm:poly}
Assume the modular arithmetic primitives used in $\mathrm{AddInpl}_G$ compile to
uniform reversible circuits of size polynomial in $m$, the modular inversion
subroutine uses a bounded number of iterations polynomial in $m$, and the
scalar-register size $n$ is polynomially bounded in $m$. Then the oracle circuit
family realizing $\mathcal{O}_{\mathrm{impl}}$ is uniform and polynomial in $m$.
\end{theorem}

Theorem~\ref{thm:poly} excludes prototype constructions that hide exponential advice in hard-coded encodings or bespoke permutations; in our setting, the cost is explicit in reversible modular arithmetic and polynomial in the register size.

\section{Case Study: Semantic and Complexity Witnesses}\label{sec:casestudy}

For the case study, we use the short Weierstrass curve $E: y^2 \equiv x^3 + ax + b \pmod p$ over the prime field $\mathbb{F}_p$, with parameters $p=5$, $a=3$, and $b=3$. We measure the oracle components from the Qrisp-based ECDLP implementation~\cite{polimeni2025end}. As in the reference code, we adopt the convention that $(0,0)$ encodes
the identity element $\mathcal{O}$.

\textbf{WF filtering and a semantic witness for Lemma~\ref{lem:addinpl}.}
Because the in-place addition gadget is derived for a generic affine case, we first filter accumulator inputs using a lightweight well-formedness screen that avoids immediate singular intermediates induced by the gadget's initial normalization steps, e.g., values that make $R_x-G_x \equiv 0 \pmod p$ or $R_y-G_y \equiv 0 \pmod p$. For this instance, the filter yields a single candidate point $R_0=(4,3)$ with constant point $G=(3,2)$. On this input, a basis-state simulation witness confirms that the in-place gadget computes the expected group update:
\[
\mathrm{AddInpl}_G(R_0)=(4,2)=R_0\oplus G,
\]
matching the classical reference implementation. This provides a concrete witness
supporting Lemma~\ref{lem:addinpl} under $\mathsf{WF}$.

\textbf{Control-law anomaly for the elliptic-curve gadget.} 
Shor-style ECDLP oracles require controlled invocations of the point-update gadget. We therefore test the standard control law: when the control qubit is set to $\ket{0}$, the controlled operation should act as the identity on the target. Let $c$ be a control qubit, and consider the externally controlled invocation of $\mathrm{AddInpl}_G$ on the accumulator register encoding $R=(4,3)$, implemented using the Qrisp control context $\mathrm{control}(c)$. With $c=\ket{1}$, the update behaves as expected and maps $(4,3)\mapsto(4,2)$. However, with $c=\ket{0}$, the accumulator changes from $(4,3)$ to $(3,3)$ in our environment, violating the expected identity behavior for a controlled operation. To rule out a global control malfunction, we additionally validate a trivial controlled-$X$ sanity test under $c=\ket{0}$, which yields deterministic measurements $\{00:256\}$, confirming that basic control works and that the anomaly is specific to the elliptic-curve gadget's controlled execution/compilation. This observation motivates treating controlled-oracle correctness as an explicit refinement obligation rather than an assumption.

\begin{table}[t]
\centering
\caption{Compilation metrics for AddInpl and Ctrl-AddInpl on the prototype instance}
\label{tab:metrics}
\begin{tabular}{lrrrr}
\hline
Circuit & \#Qubits & Depth & T-count & CX-count \\
\hline
AddInpl & 6 & 39,130 & 13,652 & 36,261 \\
Ctrl-AddInpl & 7 & 39,196 & 13,709 & 36,406 \\
\hline
\end{tabular}
\end{table}

\textbf{Compilation metrics (evidence for Lemma~\ref{lem:structcost}).} 
Table~\ref{tab:metrics} reports compilation metrics for the prototype instance. These data should be interpreted only as a concrete compilation witness showing that the studied implementation produces finite circuits for the selected parameters and that adding control introduces limited overhead in this small instance. They are not intended as empirical scaling evidence across field sizes. The polynomial-family claim in Theorem~\ref{thm:poly} instead follows from the structural form of the implementation: scalar multiplication invokes a point-addition routine once per scalar bit, and each point addition is reduced to reversible modular arithmetic over $m$-bit registers. Therefore, under the stated assumption that the underlying modular arithmetic and fixed-round inversion routines compile to uniform polynomial-size reversible circuits, the resulting oracle family is polynomial in $m$.

To summarize, our case study shows that the in-place point-update gadget matches the classical reference under $\mathsf{WF}$, but controlled execution violates the expected control law under the evaluated toolchain, indicating a semantic gap that must be treated as an explicit verification obligation. 
The case study is therefore intended as a semantic witness and bug-finding demonstration, rather than a broad empirical validation across curves or field sizes.

\section{Conclusion and Future Work}\label{sec:conclusion}

In this paper, we present a semantics-based specification of the implemented ECDLP oracle and a refinement-based framework for reasoning about its correctness and complexity. Our case study provides concrete evidence that (i) the in-place point-update gadget can match a classical reference on well-formed inputs, and (ii) controlled execution can violate the expected control law even when basic controlled-$X$ behavior is correct. Together, these results highlight the need to treat controlled-oracle correctness as an explicit verification obligation in ECDLP-oriented quantum software. 

Future work will (i) mechanize the core refinement obligations (e.g., Lemmas~\ref{lem:addinpl}--\ref{lem:multadd}) in a proof assistant or via symbolic reasoning over $\mathbb{F}_p$, (ii) extend $\mathsf{WF}$ to cover exceptional cases (e.g., point-at-infinity and inverse-point additions) using complete point representations such as projective coordinates, and (iii) establish sufficient conditions---or a verified compilation path---under which controlled oracle components satisfy the standard control law, enabling an end-to-end refinement argument for the full oracle and the associated classical post-processing.

\section*{Acknowledgments}
We thank Dr.\ Tan Teik Guan (CEO and Co-founder of pQCee~\cite{pQCeePteLtd2026Mar}) for early discussions that inspired the direction of this study.

\subsubsection{\discintname}
The authors have no competing interests to declare that are relevant to the content of this article.

\bibliographystyle{splncs04}
\bibliography{refs}

\end{document}